\documentclass[12pt]{article}
\usepackage{amsmath}

\newtheorem{theorem}{Theorem}

\newtheorem{definition}[theorem]{Definition}

\newenvironment{proof}[1][Proof]{\textbf{#1.} }{\ \rule{0.5em}{0.5em}}
\numberwithin{equation}{section}

\begin{document}

\title{Energy estimate for initial data on a characteristic cone}
\author{
Yvonne Choquet-Bruhat \\ 
Acad\'emie des Sciences, Paris \\[3mm]
Jos\'e M. Mart\'\i n-Garc\'\i a \\
Institut d'Astrophysique de Paris, and \\
Laboratoire Univers et Th\'eories, Meudon}

\maketitle

\begin{abstract}
The Einstein equations in wave map gauge are a geometric second order system
for a Lorentzian metric. To study existence of solutions of this hyperbolic
quasi diagonal system with initial data on a characteristic cone which are
not zero in a neighbourhood of the vertex one can appeal to  theorems due to
Cagnac and Dossa, proved for a scalar wave equation, for initial data in
functional spaces relevant for their proofs. It is difficult to check that
the initial data that we have constructed as solutions of the Einstein
wave-map gauge constraints satisfy the more general of the Cagnac-Dossa
hypotheses which uses weighted energy estimates. In this paper we start a
new study of energy estimates using on the cone coordinates adapted to its
null structure which are precisely the coordinates  used to solve the
constraints, following work of Rendall who considered the Cauchy problem for
Einstein equations with data on two intersecting characteristic surfaces.
\end{abstract}

\newpage

\section{Introduction}

In recent work (see summary in \cite{CBChMG2010}) we have
considered the Cauchy problem for the Einstein equations with
data on a characteristic cone. We have used a wave-map gauge
with target a Minkowski metric which admits this cone as a null
cone and derived explicit formulae for the constraint equations
on initial data, that is the trace on the cone of the
looked for Lorentzian metric. These constraints were proved to
be necessary and sufficient conditions for a solution of the
Einstein equations in wave gauge taking these initial data to
be a solution of the original Einstein equations. We have
constructed solutions of the constraints which tend to
Minkowskian values at the vertex of the cone, but are not
necessarily identical to the trace of a Minkowski metric in a
neighbourhood of this vertex, as was assumed in a recent book
by Christodoulou \cite {Christodoulou2008} and a subsequent
paper by Klainerman and Rodnianski \cite
{KlainermanRodnianski2009}. The Einstein equations in wave map
gauge are a geometric second order system for a Lorentzian
metric. To study existence of solutions of this hyperbolic
quasi diagonal system with initial data on a characteristic
cone which are not zero in a neighbourhood of the vertex we
have appealed to a theorem due to Cagnac and Dossa, proved for
a scalar wave equation, for initial data in functional spaces
relevant for their proofs. However it is difficult to check
that the initial data that we have constructed as solutions of
the Einstein wave-map gauge constraints satisfy the more
general of the Cagnac-Dossa hypotheses which appeals to
weighted energy estimates.

In this paper we start a new study of energy estimates using on the cone
coordinates adapted to its null structure which are precisely the
coordinates we used to solve the constraints, inspired by work of Rendall
\cite{Rendall1990} and Damour and Schmidt \cite{DamourSchmidt1990} who
considered the Cauchy problem for Einstein equations with data on two
intersecting characteristic surfaces.

\section{Definitions and notations}

We consider a linear quasidiagonal second order system on a manifold $V$%
\begin{equation}
g^{\alpha\beta}D_{\alpha\beta}^{2}h=f
\label{2.1}
\end{equation}
where $g$ is a Lorentzian metric, $h$ and $f$ are sections of a vector
bundle $\mathcal{V}$ over $V$ (for example covariant symmetric 2 tensor
fields) and $D$ is a covariant derivative in a given metric
not necessarily equal to $g$.

We take for $V$ an open set of $R^{n+1}$ and denote by
$y^{\alpha},\alpha=0,i,$ with $i=1,...n,$ coordinates
admissible for the differential structure of $V.$ We consider a
cone $C_{O}$ of vertex $O\in V$ which has in the coordinates
$y^{\alpha}$ the same equation as the Minkowski cone in
standard coordinates
\begin{equation}
y^{0}=r,\text{ \ \ }r^{2}:=\sum_{i=1}^{n}(y^{i})^{2};
\label{2.2}
\end{equation}
we suppose $C_{O}$ to be a characteristic cone of the
Lorentzian metric $g:$ it is well known that the use of normal
geodesic coordinates centered at the vertex $O$ shows that the
choice (\ref{2.2}) is no restriction on $g$ and $C_{O}$ if $V$
is a small enough neighbourhood of $O.$ Cagnac and
Dossa use the same representation of a characteristic cone with
vertex $O.$ They
denote, as we will do, by $Y_{O}^{T}$ the future of $O$ limited by
$y^{0}\leq T,$ that is:
\begin{equation}
Y_{0}^{T}:=\{r\leq y^{0}\leq T\}\text{ \ and set \ \ }S_{t}:=\{r\leq
y^{0}=t\},\text{ \ \ }C_{0}^{T}:=\{r=y^{0}\leq T\};
\label{2.3}
\end{equation}
they take as coordinates on $C_{O}$ the $n$ variables $y^{i}.$
Of course the
cone is not diffeomorphic to $R^{n},$ being singular for
$\overrightarrow{y}:=(y^{1},...y^{n})=0.$

We define coordinates in $V,$ singular
at $O,$ adapted to the null structure of $C_{O},$ defined by
\begin{equation}
y^{0}=x^{1}-x^{0},\text{ \ }r=x^{1},\text{ \ }y^{i}=r\Theta^{i}(x^{A}),\text{
\ with \ \ }\sum_{i=1}^{n}(\Theta^{i})^{2}=1,
\label{2.4}
\end{equation}
$x^{A},A=2,...n$ local coordinates on the sphere $S^{n-1}.$
Components of geometric objects in $y$ coordinates are
underlined, components are in $x$ coordinates if not
underlined.

In the coordinates $x^{\alpha}$ the equation of $C_{O}$ is $x^{0}=0.$
Traces on the cone are overlined, $\bar{g}^{00}\equiv 0$.
The lines $x^{A}=$constant on $C_{O}$ are geodesic null rays,
hence $\bar{g}_{11}\equiv\bar{g}_{1A}\equiv 0;$ that is, the trace
on $C_{O}$ of the metric $g$ takes the form
\begin{equation}
\bar{g}=\bar{g}_{00}(dx^{0})^{2}+2\nu_{0}dx^{0}dx^{1}+2\nu_{A}dx^{0}dx^{A}
+\tilde{g},
\text{ \ \ }\tilde{g}:=\bar{g}_{AB}dx^{A}dx^{B}.
\label{2.5}
\end{equation}

\section{Stress energy tensor}

To have norms for tensors on $V\subset R^{n+1}$ we endow it with the
euclidean metric

\begin{equation}
\mathbf{e}\equiv(dy^{0})^{2}+\sum_{i=1,...n}(dy^{i})^{2},
\label{3.1}
\end{equation}
which reads in the $x^{\alpha}$ coordinates
\begin{equation*}
\mathbf{e}\equiv(dx^{0})^{2}-2dx^{1}dx^{0}+2(dx^{1})^{2}
+(x^{1})^{2}s_{AB}dx^{A}dx^{B}.
\end{equation*}
In the $x$ coordinates it holds that
\begin{equation}
\bar{e}^{00}=2,\text{ \ }
\bar{e}^{11}=1,\text{ \ }
\bar{e}^{01}=1,\text{ \ }
\bar{e}^{1A}=\bar{e}^{0A}=0,\text{ \ }
\bar{e}^{AB}\equiv(x^{1})^{-2}s^{AB}.
\label{3.2}
\end{equation}

We denote by $D$ the covariant derivative in the metric \textbf{e} on
$R^{n+1},$ it coincides with the covariant derivative in the Minkowski
metric $\eta,$
\begin{equation}
\eta\equiv-(dy^{0})^{2}+\sum_{i=1,...n}(dy^{i})^{2}
\equiv-(dx^{0})^{2}+2dx^{1}dx^{0}+(x^{1})^{2}s_{n-1},
\label{3.3}
\end{equation}
both these covariant derivatives coinciding with ordinary partial
derivatives in the $y$ coordinates.

Indices are raised with the contravariant associate of $g.$ We
denote by an underlined dot the pointwise scalar product
relative to \textbf{e.}

\begin{definition}
The stress energy tensor of a tensor $h$ is the symmetric 2-tensor:
\begin{equation}
U^{\alpha \beta }=D^{\alpha }h\underline{.}D^{\beta }h-{\frac{1}{2}}%
g^{\alpha \beta }D_{\lambda }h\underline{.}D^{\lambda }h.  \label{3.4}
\end{equation}
\end{definition}

We consider a past oriented timelike vector $X$. The energy momentum vector
is
\begin{equation}
\mathcal{P}^{\alpha}:=U^{\alpha\beta}X_{\beta}.  \label{3.5}
\end{equation}
The \textbf{e}-divergence of $\mathcal{P}$ is
\begin{equation}
D_{\alpha}\mathcal{P}^{\alpha}\equiv D_{\alpha}(U^{\alpha\beta}X_{\beta
})\equiv X_{\beta}D_{\alpha}U^{\alpha\beta}+U^{a\beta}D_{\alpha}X_{\beta }
\label{3.6}
\end{equation}
We have
\begin{equation}
D_{\alpha}U^{\alpha\beta}\equiv g^{\alpha\lambda}D_{\alpha\lambda}^{2}h%
\underline{.}D^{\beta}h+F^{\beta} .  \label{3.7}
\end{equation}
with
\begin{equation}
F^{\beta}\equiv D_{\alpha}g^{\alpha\lambda}D_{\lambda}h\underline{.}D^{\beta
}h+D^{\alpha}h\underline{.}D_{\alpha}D^{\beta}h-{\frac{1}{2}}D_{\alpha
}(g^{\alpha\beta}g^{\lambda\mu})D_{\lambda}h\underline{.}D_{\mu}h-g^{\alpha%
\beta}D_{\alpha}D_{\lambda}h\underline{.}D^{\lambda}h.
\end{equation}
Changing ordering and names of indices we find
\begin{equation}
F^{\beta}\equiv D^{\alpha}h\underline{.}(D_{\alpha}D^{\beta}h-D^{\beta
}D_{\alpha}h)+D_{\alpha}g^{\alpha\lambda}D_{\lambda}h\underline{.}D^{\beta
}h-{\frac{1}{2}}D_{\alpha}(g^{\alpha\beta}g^{\lambda\mu})D_{\lambda }h%
\underline{.}D_{\mu}h;
\end{equation}
the nullity of the Riemann tensor of the Minkowski metric implies
\begin{equation*}
(D_{\alpha}D^{\beta}h-D^{\beta}D_{\alpha}h\equiv g^{\beta\lambda}(D_{\alpha
}D_{\lambda}h-D_{\lambda}D_{\alpha}h)+D_{\alpha}g^{\beta\lambda}D_{\lambda
}h\equiv D_{\alpha}g^{\beta\lambda}D_{\lambda}h.
\end{equation*}
Finally we see that $F^{\beta}$ reduces to the following quadratic form in
the derivatives of $g$
\begin{equation}
F^{\beta}\equiv{\frac{1}{2}}D_{\alpha}(g^{\alpha\mu}g^{\beta\lambda}+g^{%
\beta\mu}g^{\alpha\lambda}-g^{\alpha\beta}g^{\lambda\mu})D_{\lambda }h%
\underline{.}D_{\mu}h.
\end{equation}

\section{Energy equality}

We assume that the contravariant associate of $g$ and $X$ are $C^{1}$ in
$V.$ Then $\mathcal{P}^{\alpha}\equiv U^{\alpha\beta}X_{\beta}\in C^{1}$
if $h\in C^{2}.$

When $h$ is solution of the system (\ref{2.1}) we deduce from (\ref{3.6})
the equality
\begin{equation}
D_{\alpha}\mathcal{P}^{\alpha}=X_{\beta}(D^{\beta}h\underline{.}f+F^{\beta
})+U^{a\beta}D_{\alpha}X_{\beta}.
\label{4.1}
\end{equation}

We denote by $\Omega_{\mathbf{e}}$ the $n+1$ volume form of $\mathbf{e};$
it reads in arbitrary coordinates $z^{\alpha}$
\begin{equation}
\Omega_{\mathbf{e}}=(\det \mathbf{e}_{z})^{\frac{1}{2}}dz^{0}\wedge
dz^{1}...\wedge dz^{n}\ .
\label{4.2}
\end{equation}
In the coordinates respectively $y^{\alpha }$ and $x^{\alpha }$ it holds
that
\begin{equation}
(\det \mathbf{e}_{y})^{\frac{1}{2}}\equiv 1,
\text{ \ \ \ \ }
(\det \mathbf{e}_{x})^{\frac{1}{2}}\equiv (x^{1})^{n-1}|\det s_{n-1}|^{\frac{1}{2}}.
\label{4.3}
\end{equation}
We recall the identity ($d$ denotes the exterior derivative and a dot the
contraction in the metric $g$)
\begin{equation}
D.\mathcal{P}\Omega_{\mathbf{e}}\equiv d(\mathcal{P}.\omega),  \label{4.4}
\end{equation}
where $\omega$ is the covariant vector valued Leray $n$ form whose
components are given in arbitrary coordinates $z^{\alpha}$ by
\begin{equation}
\omega_{\alpha}=(-1)^{\alpha}|\det \mathbf{e}_{z}|^{\frac{1}{2}%
}dz^{0}\wedge dz^{1}...\wedge d\hat{z}^{\alpha }\wedge ...\wedge dz^{n}.
\label{4.5}
\end{equation}
The notation $\hat{\alpha}$ means that the corresponding differential does
not appear in the component $\omega_{\alpha}$.

We choose for $X$ the past oriented vector with components in the $y$
coordinates (recall that we underline such components)
\begin{equation}
\underline{X}_{\beta}:=\delta_{\beta}^{0}.  \label{4.6}
\end{equation}

We integrate with respect to the volume form $\Omega_{\mathbf{e}}$ the
equality (\ref{4.1}) on $Y_{O}^{T}$ oriented by the natural orientation of
$R^{n}$ and increasing $t:=y^{0}.$ The result reads in the $y$ coordinates
\begin{equation}
\int_{Y_{O}^{T}}D.\mathcal{P}\Omega _{\mathbf{e}}=\int_{0}^{T}\int_{S_{t}}(%
\underline{D^{0}h}\underline{.}f+\underline{F}^{0})\mu _{e}dt.  \label{4.7}
\end{equation}

On the other hand, the following identity holds if the integral on its right
hand side exists,
\begin{equation}
\int_{Y_{O}^{T}}D.\mathcal{P}\Omega_{\mathbf{e}}\equiv\int_{\partial
Y_{O}^{T}}\mathcal{P}.\omega.  \label{4.8}
\end{equation}
We have, using the definitions 2.3
\begin{equation}
\partial Y_{O}^{T}\equiv S_{T}\cup C_{O}^{T}.  \label{4.9}
\end{equation}

\subsection{Integral on $S_{T}$}

We have
\begin{equation}
\int_{S_{T}}\mathcal{P}.\omega=\int_{r\leq T}\mathcal{\underline{\mathcal{P}}%
}^{0}(T,\overset{\rightarrow}{y})dy^{1}...dy^{n},\text{ \ \ }\overset{%
\rightarrow}{y}:=(y^{1},...y^{n}).  \label{4.10}
\end{equation}
With the choice we have made of $X,$ $\underline{\mathcal{P}}^{0}$ reads
\begin{equation}
\mathcal{\underline{\mathcal{P}}}^{0}\equiv\delta_{\beta}^{0}\underline
{U}^{0\beta}=\underline{D^{0}h}\underline{.}\underline{D^{0}h}-{\frac{1}{2}}%
\underline{g^{00}}D_{\lambda}h\underline{.}D^{\lambda}h,  \label{4.11}
\end{equation}
i.e.
\begin{equation}
\underline{\mathcal{P}}^{0}\equiv\underline{D^{0}h}\underline{.}\underline{%
D^{0}h}-{\frac{1}{2}\underline{g^{00}}} \left(\underline{g_{00}}\underline{%
D^{0}h}\underline{.}\underline{D^{0}h}+ 2\underline{g}_{0i}\underline{D^{0}h}%
\underline{.}\underline{D^{i}h}+ \underline{g}_{ij}\underline{D^{i}h}%
\underline{.}\underline{D^{j}h} \right).  \label{4.12}
\end{equation}
It is a positive definite quadratic form of $\underline{Dh}$ if $g$ is a
Lorentzian metric regularly sliced
(see \cite[appendix 7]{ChoquetBruhat2008}) by $S_{t}.$

\subsection{Integral on $C_{O}^{T}$}

We write the integral on $C_{0}^{T}$ in the $x^{\alpha }$ coordinates.
Recalling that $\mathcal{\bar{P}}^{0}$ denotes the value on $C_{O}$ of the
component with index zero in the $x^{\alpha }$ coordinates of the vector $%
\mathcal{P}$, we find
\begin{equation}
\int_{C_{O}^{T}}\mathcal{P}.\omega \equiv \int_{0}^{T}\int_{S^{n-1}}\mathcal{%
\bar{P}}^{0}(x^{1})^{n-1}|\det s_{n-1}|^{\frac{1}{2}}dx^{2}...dx^{n}dx^{1}.
\label{4.13}
\end{equation}
The components of $X$ in the coordinates $x^{\alpha }$ are
\begin{equation}
X_{\alpha }:=\underline{X}_{\beta }\frac{\partial y^{\beta }}{\partial
x^{\alpha }},\text{ \ i.e. \ \ }X_{0}=-1,\text{ \ }X_{1}=1,
\text{ \ }X_{A}=0 .
\label{4.14}
\end{equation}
Hence on $C_{O}^{T}$ it holds that
\begin{equation}
\mathcal{\bar{P}}^{0}(x^{1},x^{A})\equiv -\bar{U}^{00}(x^{1},x^{A})+\bar{U}%
^{01}(x^{1},x^{A}),  \label{4.15}
\end{equation}
where, using previous notations and recalling that
\begin{equation*}
\bar{g}^{00}=\bar{g}^{0A}=0,\ \ \nu ^{0}:=\bar{g}^{01}=\frac{1}{\nu _{0}},\
\ \bar{g}^{A1}\equiv -\nu ^{0}\nu ^{A},\ \ \bar{g}^{11}\equiv -(\nu ^{0})^{2}%
\bar{g}_{00}+(\nu ^{0})^{2}\nu ^{A}\nu _{A},
\end{equation*}
\begin{equation}
\bar{U}^{00}\equiv (\nu ^{0})^{2}\bar{U}_{11},\text{ \ }\bar{U}^{01}\equiv
\nu ^{0}(\nu ^{0}\bar{U}_{01}-\nu ^{0}\nu ^{A}\bar{U}_{1A}+\bar{g}^{11}%
\bar{U}_{11}),  \label{4.16}
\end{equation}
hence
\begin{equation}
\mathcal{\bar{P}}^{0}\equiv \nu ^{0}\{(-\nu ^{0}+\bar{g}^{11})\bar{U}%
_{11}+\nu ^{0}\bar{U}_{01}-\nu ^{0}\nu ^{A}\bar{U}_{1A}\},  \label{4.17}
\end{equation}
with, since $\bar{g}_{11}=\bar{g}_{1A}=0$,
\begin{equation}
\bar{U}_{1A}\equiv D_{A}\bar{h}\underline{.}D_{1}\bar{h},\text{ \ \ }\bar{U}%
_{11}\equiv D_{1}\bar{h}\underline{.}D_{1}\bar{h}  \label{4.18}
\end{equation}
and
\begin{equation}
\bar{U}_{01}\equiv \overline{D_{0}h}\underline{.}D_{1}\bar{h}-\frac{1}{2}\nu
_{0}\,\overline{D_{\lambda }h\underline{.}D^{\lambda }h}.  \label{4.19}
\end{equation}
We have
\begin{equation*}
\overline{D_{\lambda }h\underline{.}D^{\lambda }h}\equiv 2\nu ^{0}(\overline{%
D_{0}h}\underline{.}D_{1}\bar{h}-\nu ^{A}D_{A}\bar{h}\underline{.}D_{1}\bar{h%
})+\bar{g}^{11}D_{1}\bar{h}\underline{.}D_{1}\bar{h}+\bar{g}^{AB}D_{A}\bar{h}%
\underline{.}D_{B}\bar{h},
\end{equation*}
hence the transversal derivative $\overline{D_{0}h}$ disappears in $\bar{U}%
_{01}$ which reads
\begin{equation}
\bar{U}_{01}\equiv \nu ^{A}D_{A}\bar{h}\underline{.}D_{1}\bar{h}-\frac{1}{2}%
\nu _{0}(\bar{g}^{11}D_{1}\bar{h}\underline{.}D_{1}\bar{h}+\bar{g}^{AB}D_{A}%
\bar{h}\underline{.}D_{B}\bar{h});  \label{4.20}
\end{equation}
$\mathcal{\bar{P}}^{0}$ simplifies to the quadratic form
\begin{eqnarray}
\mathcal{\bar{P}}^{0} &\equiv &\nu ^{0}\big\{(-\nu ^{0}+\bar{g}^{11})D_{1}%
\bar{h}\underline{.}D_{1}\bar{h}+\nu ^{0}\nu ^{A}D_{A}\bar{h}\underline{.}%
D_{1}\bar{h}  \label{4.22} \\
&&\qquad -\frac{1}{2}(\bar{g}^{11}D_{1}\bar{h}\underline{.}D_{1}\bar{h}+\bar{%
g}^{AB}D_{A}\bar{h}\underline{.}D_{B}\bar{h})-\nu ^{0}\nu ^{A}D_{A}\bar{h}%
\underline{.}D_{1}\bar{h}\big\},  \notag
\end{eqnarray}
which simplifies to
\begin{equation}
\mathcal{\bar{P}}^{0}\equiv -\big\{\nu ^{0}(\nu ^{0}-\frac{1}{2}\bar{g}%
^{11})D_{1}\bar{h}\underline{.}D_{1}\bar{h}+\frac{1}{2}\bar{g}^{AB}D_{A}\bar{%
h}\underline{.}D_{B}\bar{h}\big\}.  \label{4.23}
\end{equation}
We remark using the values (\ref{4.14}) of $\bar{X}_{\alpha}$ that on
$C_{O}$
\begin{equation}
\bar{g}^{\alpha \beta }\bar{X}_{\alpha }\bar{X}_{\beta }\equiv -2\nu ^{0}+%
\bar{g}^{11}  \label{4.24}
\end{equation}
which is negative if $\bar{X}$ is timelike. Hence $\mathcal{\bar{P}}^{0}\leq
0,$ as foreseen from the general theory since the boundary $C_{O}^{T}$ of $%
Y_{O}^{T}$ is null and outgoing. (See \cite[appendix 7]{ChoquetBruhat2008}.)

\subsection{Energy equality}

We have proved, under the indicated condition, the following theorem.

\begin{theorem}
If the metric $g$ is $C^{1}$ a $C^{2}$ solution of the equation (\ref{2.1})
satisfies the equality
\begin{equation*}
\int_{S_{T}}\underline{\mathcal{P}}^{0}(T,\overset{\rightarrow}{y}%
)dy^{1}...dy^{n}=-\int_{C_{O}^{T}}\mathcal{\bar{P}}%
^{0}(x^{1},x^{A})(x^{1})^{n-1}dx^{1}\mu_{S^{n-1}}+
\end{equation*}
\begin{equation}
\int_{0}^{T}\int_{S_{t}}(\underline{D^{0}h}\underline{.}f+\underline{F}%
^{0})(t,y^{1}...,y^{n})dy^{1}...dy^{n}dt .  \label{4.25}
\end{equation}
\end{theorem}

\section{Energy inequality}

The hypothesis that the Lorentzian metric $g$ is regularly sliced on
$Y_{O}^{T}$ implies that there exist numbers $C_{m}>0$ and $C_{M}\geq C_{m}$
such that, with $\overrightarrow{y}:=(y^{i},i=1,...n),$
\begin{equation}
C_{m}\varepsilon(t,\overrightarrow{y})\leq\mathcal{\underline{\mathcal{P}}}%
^{0}(t,\overrightarrow{y})\leq C_{M}\varepsilon(t,\overrightarrow {y}),
\label{5.1}
\end{equation}
with
\begin{equation}
\varepsilon(t,\overrightarrow{y})\equiv\{\frac{\partial h}{\partial t}%
\underline{.}\frac{\partial h}{\partial t}+\delta^{ij}\frac{\partial h}{%
\partial y^{i}}\underline{.}\frac{\partial h}{\partial y^{j}}\}(t,%
\overrightarrow{y}).  \label{5.2}
\end{equation}
We set
\begin{equation}
E(t)\equiv\int_{0\leq r\leq t}\varepsilon_{t}(t,\overrightarrow{y}%
)dy^{1}...dy^{n},\text{ \ }r:=\{\Sigma(y^{i})^{2}\}^{\frac{1}{2}}.
\label{5.3}
\end{equation}

We denote generically by $C$ a number depending only on $n$ and
the uniform slicing hypotheses, i.e. $C_{m}$ and $C_{M}.$ We
have
\begin{equation}
E(t)\leq C\int_{S_{t}}\mathcal{\underline{\mathcal{P}}}^{0}dy^{1}...dy^{n}.
\label{5.4}
\end{equation}

We assume that there exists a continuous function, $C_{Dg}(t),$ of
$t\in\lbrack0,T]$ \ such that
\begin{equation}
\sup_{S_{t}}|Dg|\leq C_{Dg}(t).  \label{5.5}
\end{equation}
We denote by $C_{|Dg|}$ any number depending only on the
uniform slicing bounds of $g$ and the supremum of $C_{Dg}(t)$\
for\ $0\leq t\leq T.$
\begin{theorem}
(energy inequality) If $g$ is $C^{1}$ and uniformly sliced on $Y_{O}^{T}$
any $C^{2}$ solution of the equation (\ref{2.1}) satisfies an inequality
\begin{equation*}
E(T)\leq Ce^{C_{|Dg|}T}\int_{0}^{T}\Big(||f||_{L^{2}(S_{t})}^{2}+t^{n-1}%
\int_{S^{n-1}}|\mathcal{\bar{P}}^{0}(t,x^{A})|\mu_{S^{n-1}}\Big)dt.
\end{equation*}
\end{theorem}

\begin{proof}
We deduce from (\ref{5.1}),(\ref{5.2}) that we have on $S_{t}$%
\begin{equation}
|D^{0}u\underline{.}f+\underline{F}^{0}|\leq C\varepsilon(t)^{\frac{1}{2}%
}|f|+C_{|Dg|}\varepsilon(t).  \label{5.6}
\end{equation}
On the other hand we have
\begin{equation}
-\int_{C_{O}^{T}}\mathcal{\bar{P}}^{0}(x^{1},x^{A})(x^{1})^{n-1}dx^{1}%
\mu_{S^{n-1}}\equiv\int_{0}^{T}\Phi_{t}dt,  \label{5.7}
\end{equation}
with (recall that $x^{1}=t$ on $C_{O})$
\begin{equation}
\Phi_{t}:=-t^{n-1}\int_{S^{n-1}}\mathcal{\bar{P}}^{0}(t,x^{A})\mu_{S^{n-1}}%
\geq0.  \label{5.8}
\end{equation}
The equality (\ref{5.7}) implies the inequality
\begin{equation}
E(T)\leq\int_{0}^{T}\{C_{|Dg|}E(t)\}dt+C%
\int_{0}^{T}(||f||_{L^{2}(S_{t})}^{2}+\Phi_{t})dt,  \label{5.9}
\end{equation}
with
\begin{equation*}
||f||_{L^{2}(S_{t})}^{2}:=\int_{0\leq r\leq t}|f(t,\overrightarrow{y}%
)|^{2}dy^{1}...dy^{n}.
\end{equation*}

By the Gronwall lemma the inequality (\ref{5.9}) verified by a $C^{2}$
solution of (\ref{2.1}) such that $E(0)=0$ implies that
\begin{equation*}
E(t)\leq Z(t),
\end{equation*}
with $Z(t)$ solution of the differential equation
\begin{equation}
Z^{\prime}(t)=C_{|Dg|}Z(t)+C(||f||_{L^{2}(S_{t})}^{2}+\Phi_{t}\text{ })
\text{ \ \ with \ \ }Z(0)=0.
\label{2.10}
\end{equation}
We look for a solution of (\ref{2.10}) vanishing for $t=0$ under the form
$Z(t)=ke^{C_{|Dg|}t},$ we find
\begin{equation*}
k^{\prime}e^{C_{|Dg|}t}=C(||f||_{L^{2}(S_{t})}^{2}+\Phi_{t}),
\end{equation*}
\begin{equation*}
Z(t)\equiv e^{C_{|Dg|}t}
\int_{0}^{t}e^{-C_{|Dg|}t}C(||f||_{L^{2}(S_{t})}^{2}+\Phi_{t})dt.
\end{equation*}
\end{proof}

\subsection{Uniqueness theorem}

A uniqueness theorem for the linear equation (\ref{2.1}) results immediately
from the inequality (\ref{5.9}) which implies $E(t)\equiv0$ if
$f\equiv\Phi_{t}=0$. We state:

\begin{theorem}
Two $C^{2}$ solutions of the equation (\ref{2.1}) in $Y_{O}^{T}$ with $g$ a
$C^{1}$ Lorentzian metric uniformly sliced coincide in $Y_{O}^{T}$
if they have the same trace on $C_{O}^{T}.$
\end{theorem}

\section{Open problems}

The energy inequality can very likely be extended to tensors which are in
spaces obtained by completion of $C^{2}$ using norms which appear in this
inequality.

One could perhaps, using the energy inequality and some functional analysis,
prove an existence theorem for a generalized solution of the linear system,
as one does for a linear system with spacelike Cauchy data, though one
should probably for such a proof use a double null foliation, like in
Klainerman and Nicolo \cite{KlainermanNicolo2002}.

Anyway one needs higher order estimates to have results in the
case of quasilinear equations. In the case of a cone as support
of the initial data a problem for the use of standard embedding
and multiplication properties of Sobolev spaces is that the
sections $S_{t}$ cannot be considered as Riemannian manifolds
with equivalent Sobolev constants when $t$ tends to zero, the
vertex of the cone. A remedy proposed by Dossa in the case of a
scalar equation is to use the $y^{i}$ as coordinates on the
cone and to scale $\overrightarrow{y}$ by powers of $t^{-1}$ in
order to work in a fixed sphere of $R^{n}.$ We postpone the
application of this idea to a further work.


\begin{thebibliography}{99}
\bibitem{Cagnac1981} F.~Cagnac, Annali di Matematica \textbf{129} (1981),
13--41.

\bibitem{ChoquetBruhat2008} Y.~Choquet-Bruhat, \emph{General relativity
and the Einstein equations}, Oxford University Press, 2008.

\bibitem{CBChMG2010} Y.~Choquet-Bruhat, P.T. Chru\'{s}ciel, and J.M.
Mart\'{i}n-Garc\'{i}a, to appear World Scientific, arXiv: 1002.1471v1 [gr-qc].

\bibitem{Christodoulou2008} D.~Christodoulou, \emph{The formation of black
holes in general relativity}, EMS, 2008.

\bibitem{DamourSchmidt1990} T.~Damour and B.~Schmidt, Jour.\ Math.\ Phys.
\textbf{31} (1990), 2441--2453.

\bibitem{Dossa1997} M.~Dossa, Ann. Inst. H. Poincar\'{e} Phys. Th\'{e}or.
(1997), 37--107.

\bibitem{KlainermanNicolo2002} S. Klainerman and F. Nicolo
\emph{The evolution problem in General Relativity}, Birkha\"{u}ser 2002.

\bibitem{KlainermanRodnianski2009} S. Klainerman and I.\ Rodnianski,
arXiv: 0912.5097v1 [gr-qc].

\bibitem{Leray1953} J.~Leray, \emph{Hyperbolic differential equations},
notes, Princeton (1953).

\bibitem{Rendall1990} A.D. Rendall, Proc.\ Roy.\ Soc.\ London A \textbf{427}
(1990), 221--239.
\end{thebibliography}
\end{document}